\documentclass[%
reprint,
superscriptaddress,
amsmath,amssymb,
aps,
prx,
]{revtex4-2}

\usepackage{makecell}
\usepackage{comment}
\usepackage[utf8]{inputenc}
\usepackage[T1]{fontenc}
\usepackage{graphicx}
\usepackage{dcolumn}
\usepackage{appendix}
\usepackage{amsmath}
\usepackage{bm}
\usepackage{bbm} 
\usepackage{amsthm}

\usepackage{braket}
\usepackage{booktabs}
\usepackage{bbold}
\usepackage{mathtools}
\usepackage{mathtools}
\usepackage{physics}
\usepackage[colorlinks]{hyperref}
\usepackage{xcolor}
\newtheorem{prop}{Proposition}
\newtheorem{conj}{Conjecture}
\newtheorem{Theorem}{Theorem}

\newcommand{\raf}[1]{{\color{green}#1}}

\begin{document}

\preprint{APS/123-QED}

\title{Bell inequalities tailored to optimal global randomness certification}

\author{Ignacio Perito}
\email{npnacho@gmail.com}
\affiliation{ICFO-Institut de Ciencies Fotoniques, The Barcelona Institute of
Science and Technology, 08860 Castelldefels (Barcelona), Spain}

\author{Raffaele D'Avino}

\affiliation{ICFO-Institut de Ciencies Fotoniques, The Barcelona Institute of
Science and Technology, 08860 Castelldefels (Barcelona), Spain}

\author{Michał Jung}
\affiliation{Center for Quantum-Enabled Computing, Center for Theoretical Physics,
Polish Academy of Sciences, al. Lotników 32/46, 02-668 Warsaw, Poland}

\author{Piotr Mironowicz}

\affiliation{Center for Theoretical Physics, Polish Academy of Sciences, al. Lotnik\'{o}w 32/46, 02-668 Warsaw, Poland}
\affiliation{Department of Algorithms and System Modeling, Faculty of Electronics, Telecommunications and Informatics, Gda\'{n}sk University of Technology, Poland}

\author{Antonio Ac\'in}

\affiliation{ICFO-Institut de Ciencies Fotoniques, The Barcelona Institute of
Science and Technology, 08860 Castelldefels (Barcelona), Spain}
\affiliation{ICREA-Institucio Catalana de Recerca i Estudis Avançats, Lluis Companys 23, 08010 Barcelona, Spain}

\author{Remigiusz Augusiak}
\affiliation{Center for Quantum-Enabled Computing, Center for Theoretical Physics,
Polish Academy of Sciences, al. Lotników 32/46, 02-668 Warsaw, Poland}


\begin{abstract}
    We present two novel families of bipartite Bell inequalities designed to achieve optimal global randomness certification for an arbitrary number of outputs $d$. We first use symmetry arguments to argue that their maximal quantum violations certify $2\log d$ random bits. For the first family, we construct a quantum realization using $d\times d$ maximally entangled states which provides a quantum violation that we conjecture to be optimal for any $d$. It is then numerically shown that the obtained quantum violation certifies optimal global randomness, up to numerical precision, for $d=3,4$.  For the second family, we provide the optimal quantum violation and its quantum realization for any $d$, again using $d\times d$ maximally entangled states and projective measurements over at least two unbiased bases on one of the parties. We self-test this realization for $d=3$, which implies the optimal certification of two fully random trits. 
\end{abstract}

\maketitle

\section{Introduction}

Randomness is not only a quantity of fundamental scientific interest but also a key resource in many security related tasks. For this reason, one of the most exploited features of Bell nonlocality \cite{bell1964} is its capacity to certify randomness in a device-independent manner. While the inherent randomness of a process is impossible to prove mathematically \cite{kolmogorovbook}, the observation of Bell nonlocal correlations between two devices, or parties, Alice and Bob, certifies the presence of genuine randomness, assuming only the validity of quantum theory. The measured outcomes are certified to be unpredictable, even by any eavesdropper, Eve, correlated with the Bell setup~\cite{AMNature2016}. A paradigmatic example of this situation is given by the well known Clauser-Horne-Shimony-Holt (CHSH) functional~\cite{chsh}, whose maximal quantum violation, attained by a two-qubit maximally entangled state, ensures that each of the local outputs is a uniform random bit. CHSH achieves optimal \emph{local randomness} certification, where at least one of the outcomes of a Bell test is fully random (in fact, for the CHSH maximal quantum violation, all of them are). Bell tests with binary outcomes attaining optimal local randomness certification for arbitrary two-qubit entangled states were later derived in~\cite{AMPPRL2012}. 

With the advance of quantum technologies, real-world platforms that exploit quantum systems of more than two levels are becoming more and more common \cite{Ringbauer2022-sb,Chi2023-za}. In consonance with this, theoretical developments have been done towards randomness certification in Bell scenarios beyond the binary outcome case. Bell tests for optimal local randomness certification of $\log d$ random bits from projective measurements were 
extended to arbitrary outcome number $d$ in~\cite{Sarkar2021SelfTestingQudits}. Then protocols based on general non-projective measurements for certification of $2\log d$ random bits were extended first to qutrit systems \cite{Borkala2022-vd} and then in \cite{farkas2026maximal} to systems of arbitrary dimension $d$~\footnote{For simplicity, in this work we use $d$ to denote the number of outcomes in the Bell test and the dimension of the considered quantum systems.}.

All these constructions certify optimal local randomness. Outputs in a Bell test, however, are in general correlated and, therefore, when taken together, they are not fully random. Attaining optimal \emph{global randomness} certification, where outputs by Alice and Bob are fully random even when combined, is interesting not only from a fundamental point of view, but also to optimize the rate of random-number generation. Designing Bell tests  to achieve this task turns out much harder than in the local case. There exist constructions for binary outcomes for bipartite systems using two-qubit maximally entangled states, e.g.~\cite{piotr2013,Law_2014}, or any entangled pure state~\cite{PhysRevResearch.2.042028}, and in an arbitrary multipartite scenarios~\cite{DelatorrePRL2015}. Beyond the binary-output case, no systematic construction exists.

In this work, we construct two families of Bell inequalities for $d$-outcome measurements and conjecture, using symmetry arguments introduced in~\cite{dhara2013}, that they attain global randomness certification for any $d$. That is, given their maximal quantum violation, it is possible to identify a measurement by Alice and by Bob whose outputs define $2\log d$ random bits. The first family is simpler as it involves 2 measurements by Alice and $d+1$ by Bob. Optimal global randomness certification is shown for $d=3,4$, up to numerical precision, using the Navascues-Pironio-Acin (NPA) hierarchy~\cite{NPA07PRL,NPA08NJP}. The second construction is slightly more complex, now involving $d$, $d+1$, measurements by Alice, Bob, and optimal global randomness certification is analytically proven for $d=3$. Beyond randomness, for the first family of Bell inequalities we construct a quantum realization conjectured to provide the optimal quantum violation, which is confirmed using NPA up to $d=6$. For the second family of inequalities, we analytically derive the maximal quantum violation and provide a quantum realization for all $d$. 

In an accompanying paper \cite{Barcelona2}, we show that our first family of Bell inequalities exhibits more robustness to noise than all known alternatives when it comes to surpassing the two-bit threshold given by the optimal projective qubit strategies.

\section{Scenario, randomness and symmetry arguments}\label{sec:scenario}

We consider a standard Bell scenario where two parties, Alice and Bob, have access to $m_a$ and $m_b$ measurements settings respectively and all those measurements have $d$ possible outcomes~\footnote{Without losing generality, the number of outcomes can be taken equal for all measurements by Alice and Bob, as if some measurements have less outputs than others, one can always pick $d$ as the maximum and assign a zero probability to some outcomes.}. We refer to this setup as the $(m_a,m_b;d,d)$ scenario. The observed correlations are fully determined by a set of probability distributions $\{p(ab|xy)\}$, where $p(ab|xy)$ is the probability that the devices emit outputs $a$ and $b$, with $a,b \in \{0,\dots,d-1\}$, when Alice chooses measurement setting $x\in \{0,\dots,m_a-1\}$ and Bob chooses $y\in \{0,\dots,m_b-1\}$). We assume quantum mechanics to be a valid (and complete) description of the setup, so there is always a choice of a shared state $\ket{\psi}$ and local measurements operators $\{M_{a|x} \}$ and $\{M_{b|y} \}$ such that the observed correlations, or behaviors, have the form:
\begin{equation}
    p(ab|xy) = \bra{\psi} M_{a|x} \otimes M_{b|y} \ket{\psi}.
\end{equation}
We refer to the set of all possible behaviors of this form as the quantum set of correlations $\mathcal{Q}$. Other relevant sets are the local polytope $\mathcal{L}$ and the non-signaling polytope $\mathcal{NS}$, which consist of behaviors that admit local hidden-variable models and, respectively, satisfy the non-signaling principle. Since the local set is a polytope, one can witness nonlocality through violation of Bell inequalities. Specifically, a behavior is nonlocal if it violates an inequality of the form $\mathcal{I} \left[ p(ab|xy) \right] \leq \beta_\mathcal{L}$, where $\mathcal{I}$ is a linear functional of the probabilities $p(ab|xy)$ and $\beta_\mathcal{L}$ is its maximum value over the local set.

In the following, it will be useful to work with the Fourier transform of the behaviors:
\begin{equation}
  \langle  A_x^{(k)} B_y^{(l)} \rangle= \sum_{ab} \omega^{ak+bl} p(ab|xy) ,
\label{eq:ft}
\end{equation}
where $\omega = \exp\left( 2\pi i/d \right) $. This notation is motivated by the fact that, for a quantum behavior, we can always map the measurements to unitary operators: $A_x = \sum_a \omega^a \Pi_{a|x}$ and $B_y = \sum_b \omega^b \Pi_{b|y}$, where $\Pi_{a|x}$ and $\Pi_{b|y}$ are the projectors onto the eigenspaces of Alice's and Bob's outcomes for measurement choices $x$ and $y$, respectively. Therefore:
\begin{equation}
  \bra{\psi} A_x^k \otimes B_y ^l \ket{\psi} = \sum_{ab} \omega^{ak+bl} p(ab|xy) = \langle  A_x^{(k)} B_y^{(l)} \rangle.
\end{equation}

In this work, we are interested in finding settings $x,y$ for which the outputs are certified to be maximally random. A necessary condition is that $p(ab|xy) = 1/d^2 \; \forall a,b$. It follows from \eqref{eq:ft} that this is the case if and only if $\langle A_x^k B_y^l \rangle = 0$ for any combination of $k$ and $l$ except $k=l=0$ for which $\langle A_x^0 B_y^0 \rangle = \sum_{ab}p(ab|xy) = 1$. 

Of course, the condition $p(ab|xy)=1/d^2$ is not enough to certify genuine and secure randomness, as it only describes the statistics observed by Alice and Bob, which may still appear non-random to a third party correlated with the setup. The outcomes of the measurements are truly random if and only if they cannot be predicted with probability one by any external party. To quantify the amount of randomness present in the measurements' outcomes, one defines the guessing probability for an eavesdropper, $G(AB|x,y)$, as the probability of Eve guessing the outputs $a,b$, when settings $x,y$ are chosen, optimized over all possible attacks compatible with Alice's and Bob's observed correlations $p(ab|xy)$, see~\cite{TuraPRL2025} and Appendix \ref{app:randsdp}, Eq.~\eqref{eq:global_guessing_probability}~\cite{nieto2014using,bancal2014more}. As it turns out, one can obtain upper bounds for $G(AB|x,y)$ by means of semi-definite programming~\cite{mironowicz2024semi}, so one can state certificates of randomness in an efficient manner.

Our approach to derive Bell inequalities for maximal global randomness certification is based on a symmetry argument presented in~\cite{dhara2013}. It is often the case, but not always, that the correlations attaining the maximal quantum violation of a Bell inequality are unique. Given a Bell inequality, and under the assumption that its maximal quantum violation is unique, it is easy to conclude that it certifies maximal randomness whenever it is invariant under certain symmetry transformations. To illustrate this idea, we start with the inequality considered in~\cite{piotr2013,Law_2014,gallego2010}
, which in fact represents the first case, $d=2$, of our general constructions for arbitrary $d$.

\subsection{An example: the $d=2$ case}

Consider the Bell operator defined in Refs. \cite{piotr2013,Law_2014,gallego2010}:
\begin{equation}
    \bar{\mathcal{I}}_2 \equiv  A_0\otimes \left( B_0 + B_1 \right)  +  A_1\otimes \left( B_0 + \omega B_1  \right)   + A_0\otimes B_2  \, ,
    \label{eq:I2}
\end{equation}
where all observables satisfy $A_x^2=B_y^2=\mathbbm{1}$.
Note that, as $\omega=-1$ in this case, the first two terms form the CHSH operator \cite{chsh} for settings $x,y \in \{0, 1 \}^2$. Now, consider the transformation $\mathcal{T}$ that maps $A_1 \rightarrow \omega A_1$ and $B_0 \leftrightarrow B_1$, that is, it permutes the outputs of Alice's second measurement and Bob's measurement settings. It is direct to see that this transformation leaves the Bell functional $\langle \bar{\mathcal{I}}_2 \rangle$ invariant, while it maps $\langle A_1 B_2 \rangle \rightarrow \omega \langle A_1 B_2 \rangle $ and $\langle A_1 \rangle \rightarrow \omega \langle A_1 \rangle $. We can also consider another transformation $\mathcal{T}'$, which just swaps the outputs of all Alice's and Bob's measurements: $A_x \rightarrow \omega A_x$ and $B_y \rightarrow \omega B_y$. This transformation also leaves $\langle \bar{\mathcal{I}}_2 \rangle$ invariant and, among other things, maps $\langle B_2 \rangle \rightarrow \omega \langle B_2 \rangle $.

Now, if the maximal violation of  $\langle \bar{\mathcal{I}}_2 \rangle$ is unique, the previous symmetry transformations imply that $\langle A_1 B_2 \rangle = \omega \langle A_1 B_2 \rangle $ for all combinations of $k$ and $l$ but $k=l=0$, which, since $\omega=-1$, automatically implies that all these correlators are zero and two bits of global randomness are certified for measurement settings $x=1$ and $y=2$.

In this particular case, neither symmetry arguments nor the uniqueness assumption is needed, given that the optimal global randomness follows from the self-testing property of the CHSH inequality \cite{Cirelson1980QuantumGeneralizations,McKague_2012}. The maximal quantum violation of $\langle\bar{\mathcal{I}}_2\rangle$ is upper bounded by $2\sqrt 2 +1$, because of the Tsirelson bound for CHSH and $\langle A_0\otimes B_2 \rangle\leq 1$. To attain this bound (i) the maximal quantum of CHSH has to achieved, which forces the state to be maximally entangled and $A_0$, $A_1$ to be projective measurements over a pair of mutually unbiased bases, say $A_0=Z$ and $A_1=X$; (ii) $\langle A_0\otimes B_2 \rangle= 1$, which implies that $B_2$ should be perfectly correlated with $A_0$, say $B_2=Z$. It is now easy to see that outputs of $A_1$ and $B_2$ define two fully random bits. Below, we build over these ideas to derive our general construction.

\section{First family of Bell inequalities}\label{sec:1fam}

For simplicity, we start by presenting the $d=3$ case,
and then move to the general case.

\subsection{The $d=3$ case}

We generalize the previous Bell functional as follows. We consider the $(2,4;3,3)$ scenario and define:
\begin{eqnarray}
    \bar{\mathcal{I}}_3&\!\!\equiv\!\!&  \frac{1}{2} \left[  A_0\otimes \left( B_0 + B_1 +B_2\right) +  A_0\otimes B_3\right.\nonumber\\
     &&\left. \hspace{0.4cm}+  A_1\otimes \left( B_0 + \omega B_1 + \omega^2  B_2 \right)    \right] + \text{h.c.,}
    \label{eq:i3bar}
\end{eqnarray}
%
where the global factor of $1/2$ is just added for later convenience. If we apply the transformation $\mathcal{T}$ that maps $A_1 \rightarrow \omega A_1$ and cyclically permute Bob's first $3$ measurements $B_y \rightarrow B_{y \bigoplus_3 1}$, it is easy to see that $ \bar{\mathcal{I}}_3 $ remains invariant. Note that $\mathcal{T}$ maps $\langle A_1^k \rangle \rightarrow \omega^k \langle A_1^k \rangle $ and $\langle A_1^k B_3^l \rangle \rightarrow \omega^k \langle A_1^k B_3^l \rangle $. On the other hand, $\mathcal{T}'$, defined as $A_x \rightarrow \omega^2 A_x$ and $B_y \rightarrow \omega B_y$ for all settings $x$ and $y$, leaves $ \bar{\mathcal{I}}_3 $ invariant and maps $\langle B_3^l \rangle \rightarrow \omega^l \langle B_3^l \rangle $. From the previous observations, we conclude that if the maximal violation happens to be unique, all non-trivial correlators are zero and then optimal randomness ($2\log(3)$ bits of randomness) is certified for measurement settings $x=1$ and $y=3$. While we cannot prove analytically the uniqueness of the maximal quantum violation of $\bar{\mathcal{I}}_3$, we show in what follows that it certifies two random trits for settings $x=1$ and $y=3$, up to numerical precision, using NPA. For that, it is convenient to study the maximal quantum violation of $\bar{\mathcal{I}}_3$.

First, note that $\bar{\mathcal{I}}_3$ can be thought of as the expanded version of a simpler Bell operator:

\begin{eqnarray}
    \mathcal{I}_3&\equiv & \frac{1}{2} \left[  A_0 \otimes \left( B_0 + B_1 +B_2\right) \right.\nonumber\\ 
    &&\left.\hspace{0.4cm}+    A_1 \otimes \left( B_0 + \omega B_1 + \omega^2  B_2 \right)  \right] + \text{h.c.,}
    \label{eq:i3}
\end{eqnarray}
defined in the $(2,3;3,3)$ scenario. Note that $\mathcal{I}_3$ has the same structure as the Buhrman-Massar functional \cite{bm}, but removing Alice's third measurement. Finding the optimal quantum value of $\langle\mathcal{I}_3\rangle$ automatically gives an upper bound for the optimal value of $\langle\bar{\mathcal{I}}_3\rangle$, because $\langle\bar{\mathcal{I}}_3\rangle\leq \langle\mathcal{I}_3\rangle+1$. One can analytically derive the following upper bound: $\langle\mathcal{I}_3\rangle \leq 3\sqrt2 \simeq 4.243$ (see Appendix \ref{app:ibounds} for details). However, this upper bound is not attainable, as it can be seen by optimizing the functional over the $1+AB$ level of the NPA hierarchy \cite{NPA07PRL,NPA08NJP}, which gives the tighter bound $\langle \mathcal{I}_3\rangle \leq 4$ (up to numerical precision), which turns out to be tight, as it can be reached by a see-saw optimization~\cite{seesawWernerWolf2001,seesawPalVertesi2010} over quantum strategies.

To find a quantum strategy attaining the bound $\langle\mathcal{I}_3\rangle =4$ one could simply resort to the output of the see-saw algorithm. However, it is instructive for the subsequent discussion to derive it by taking inspiration out of the qubit case and the intuition underlying our construction. We take a maximally entangled state, $ \ket{\psi_3} = (1/\sqrt3) \sum_k \ket{kk}$, and Alice's measurements as the generalized $Z$ and $X$ operators: 
\begin{equation}
        A_0 = Z = \sum_{k=0}^2 \omega ^ k \ket{k}\!\!\bra{k},\quad
        A_1 = X = \sum_{k=0}^2 \ket{k+1}\!\!\bra{k} .
\end{equation}
Note that they give two unbiased bases. Then, we choose Bob's measurements in such a way that:
\begin{equation}
B_0 + B_1 +   B_2 = \alpha Z^2, \quad B_0 + \omega B_1 + \omega^2  B_2 = \alpha X, 
    \label{eq:bobsopd3}
\end{equation}
where $\alpha$ is a constant that we aim to maximize. As the state $ \ket{\psi_3}$ is stabilized (it is an eigenstate with eigenvalue $1$) by both $Z \otimes Z^2$ and $X \otimes X$, we have that this strategy will saturate the quantum bound $\langle\mathcal{I}_3\rangle = 4$ iff $\alpha = 2$. It turns out that such a choice of Bob's measurements is possible and given by:
\begin{equation}
    B_y = \frac{1}{3} \left( 2Z^2 + 2\omega^{2y} X - \omega^{y+1} X^2Z \right)
    \label{eq:bobsalld3}
\end{equation} 
(see Appendix \ref{app:optstrat} for details). From here, and given that the state is maximally entangled, it is easy to find an optimal strategy for $\bar{\mathcal{I}}_3$: we take the previous strategy and add $B_3 = Z^2$ so that $\langle A_0 B_3 \rangle =1$, giving $\bar{\mathcal{I}}_3=5$. Note that this forces $A_1$ and $B_3$ to be unbiased and, given that the state is maximally entangled, we get: $p(ab|x=1,y=3)=1/9 \,\, \forall a,b$. We emphasize that this quantum realization is the analogous of that obtained for the qubit case, in the sense that the parties share the maximally entangled state of two qutrits while Alice performs two measurements corresponding to two unbiased bases and Bob aligns his last measurement with one of Alice's.

With the derived strategy, attaining the optimal quantum violation $\langle \bar{\mathcal{I}}_3\rangle=5$,
we compute the device-independent guessing probability following the numerical methods depicted in Appendix \ref{app:randsdp}, getting $G(AB|x=1,y=3)=1/9$ up to numerical precision. Then, our inequality achieves the optimal global randomness certification in the $(2,4;3,3)$ scenario.

\subsection{The general case}

The results presented before can be generalized to the $(2,d+1,d)$ scenario by defining:
\begin{equation}
    \bar{\mathcal{I}}_d \equiv \frac{1}{2} \left[  \left( \sum_{x=0}^1\sum_{y=0}^{d-1} \omega^{xy}A_x \otimes B_y \right)  +  A_0 \otimes B_d
    \right]  + \text{h.c.}     
    \label{eq:rand_ineq_gen}
\end{equation}
This expression is invariant under both the transformation $\mathcal{T}$ that maps $A_1 \rightarrow \omega A_1$ and cyclically permute Bob's first $d$ measurements $B_y \rightarrow B_{y \bigoplus_d 1}$, and the transformation $\mathcal{T'}$ defined as $A_x \rightarrow \omega^{d-1} A_x$ and $B_y \rightarrow \omega B_y$ for all settings $x$ and $y$. As above, assuming the uniqueness of the maximal quantum violation, these symmetries imply that the outcomes of settings $x=1$ and $y=d+1$ certify $2\log d$ bits of randomness. In what follows, we provide different analytical and numerical arguments supporting the fact that the maximal quantum violation of $\bar{\mathcal{I}}_d$ certifies optimal global randomness.

As before, we can start by moving to the simpler $(2,d;d,d)$ scenario by removing Bob's last measurement and define:
\begin{equation}
    \mathcal{I}_d \equiv \frac{1}{2}\sum_{x=0}^1\sum_{y=0}^{d-1} \omega^{xy}A_x \otimes B_y    + \text{h.c.,}
    \label{eq:red_ineq_gen}
\end{equation}
to ease the study of the classical, quantum and non-signaling bounds. We state here the classical bound. 
\begin{prop}
The classical bound is given by:
\begin{equation}
           \beta_\mathcal{L}(d) = 2\left[ 1 + 2 \sum_{\ell=1}^{\lfloor \frac{d+1}{2}\rfloor - 1} \cos\left( \frac{\pi \ell }{d} \right) f_\ell \right] . 
\end{equation}
where $ f_\ell = \cos (\pi/d)$ if $\ell$ is odd and $f_\ell = 1$ otherwise.
\end{prop}
Expressions for the non-signalling bound and an analytical upper bound on the quantum value are provided in the Appendix~\ref{app:ibounds}.  Now, we will see that a generalization of the strategy used in the $d=2$ and $d=3$ cases will give us the optimal quantum values for $d=4$, $d=5$ and $d=6$. We conjecture that this is the case for $d\geq7$ as well.

As before, we start by fixing the state to be $ \ket{\psi_d} = (1/\sqrt d) \sum_i \ket{ii}$ and Alice's measurements as:
\begin{equation}
        A_0 = Z = \sum_k^{d-1} \omega ^ k \ket{k}\!\!\bra{k} , \;
        A_1 = X = \sum_k^{d-1} \ket{k+1}\!\!\bra{k}\,.
\end{equation}
We demand Bob's measurements to satisfy:
\begin{equation}
       \sum_y B_y  = \alpha_d Z^\dagger, \quad
       \sum_y \omega^y B_y = \alpha_d X , 
    \label{eq:bobsopsd}
\end{equation}
with $\alpha_d$ as large as possible. It turns out that an optimal construction achieving this is given by:
\begin{equation}
    B_y = \sum_{k=0}^{d-1} \lambda_{y,k} X^{k+1} Z^k ,
    \label{eq:optbob}
\end{equation}
where the coefficients $\lambda_{y,k}$ can be found in Appendix \ref{app:optimalconstruction}. This choice of coefficients gives rise to valid measurements (that is, $B_y$ are unitary and with eigenvalues equal to the $d$ roots of unity for all $y)$. From the definition of $\lambda_{y,k}$, it can be seen that:
\begin{equation}
    \sum_y \omega^{my} \lambda_{y,k} = d \lambda_{0,k} \delta_{k,m-1},
    \label{eq:relationslambda}
    \end{equation}
where $\delta_{a,b}$ is Krönecker's delta. Particularizing at $m=0$ and $m=1$, we see that relations \eqref{eq:bobsopsd} hold with $\alpha_d=\left[ \sin \left(\frac{\pi}{2d} \right)  \right] ^{-1}$, so this strategy gives:
\begin{equation}\label{eq:tsir_bound}
    \mathcal{I}_d = 2 \left[ \sin \left(\frac{\pi}{2d} \right)  \right] ^{-1} .
\end{equation}

We conjecture that this realization provides the maximal violation of $\mathcal{I}_d$.
\begin{conj}
The Tsirelson bound of $\mathcal{I}_d$ is given by Eq.~\eqref{eq:tsir_bound} 
and is attained by taking $A_0=Z$, $A_1=X$; the state $\ket{\psi_d} = \frac{1}{\sqrt d} \sum_k \ket{kk}$; and $B_y$ defined by Eq. \eqref{eq:optbob}.
\label{conj:i}
\end{conj}

For $d=2$ and $d=3$ we recover $\langle \mathcal{I}_2 \rangle = 2\sqrt2$ and $\langle \mathcal{I}_3 \rangle = 4$. Using NPA, we can numerically see that the conjecture holds for $d=4,5,6$, up to numerical precision (see Appendix \ref{app:numericalbounds} for details). It is worth noting that the level of NPA required to see that $\langle \mathcal{I}_d \rangle$ is at most equal to the conjectured maximal violation increases with $d$, which explains why we could not go beyond $d=6$. 

For low dimensions, this inequality is notably robust and, in particular, it allows to certify nonlocality even in presence of $25\%$ of depolarizing noise for $d=3$ (see Appendix \ref{app:ratio} for details).

For $\bar{\mathcal{I}}_d$, the above strategy gives $\bar{\mathcal{I}}_d =  2 \left[ d\sin \left(\frac{\pi}{2d} \right)  \right] ^{-1} + 1 $ after setting $B_d = Z^\dagger$. Using NPA, we have checked that this violation certifies optimal global randomness for $d=3$ and $d=4$; and we find promising bounds for $d=5$ (see Appendix \ref{app:optrand34}). This, together with the symmetries present in the functional, leads us to conjecture that:
\begin{conj}
The maximal violation of $\langle\bar{\mathcal{I}}_d\rangle$, which is equal to $2 \left[ d\sin \left(\frac{\pi}{2d} \right)  \right] ^{-1} + 1$ certifies $2\log d$ bits of randomness for settings $x=1$ and $y=d$.
\label{conj:ibar}
\end{conj}

An analysis of the robustness of the randomness certification protocol in $d=3$ is detailed in Appendix \ref{app:rob}.

\section{Second family of Bell inequalities}\label{sec:2fam}

Using equations \eqref{eq:relationslambda}, it is easy to see that the optimal Bob's operators \eqref{eq:optbob} from the previous family of Bell functionals satisfy the following Fourier relations:
\begin{equation}
    \sum_y \omega ^ {my} B_y  = d\lambda_{0,m-1} X^mZ^{m-1} .
    \label{eq:ftbobsop}
\end{equation}
The optimal construction presented in the last section took advantage of setting $A_0 = Z$ and $A_1 = X$ so that their respective tensor product with the $m=0$ and $m=1$ instances of \eqref{eq:ftbobsop} stabilizes the maximally entangled state. This suggests that we can define a new functional in the $(d,d;d,d)$ scenario, whose optimal strategy extends that of $\mathcal{I}_d$ by adding $d-2$ measurements on Alice's side, each stabilizing the maximally entangled state together with the operators in \eqref{eq:ftbobsop} for $m=2,\dots,d-1$. More concretely, we define:
\begin{equation}
    \mathcal{F}_d = \frac{1}{2}\sum_{x} \lambda_{0,x-1}^*    A_x\otimes \sum_y \omega^{xy} B_y  + \text{h.c.,}
    \label{eq:bell_general_expr}
\end{equation}
where the factor of $\frac{1}{2}$ is included just for convenience. Note that picking Bob's operators as \eqref{eq:optbob}, Alice's operators as $A_m = \left( X ^m Z^{m-1} \right)$, and the state $\ket{\psi}$ we get:
\begin{equation}
A_m \otimes \sum_y \omega^{my} B_y =  \left( X ^m Z^{m-1} \right)^* \otimes \left( X ^m Z^{m-1} \right) , 
\end{equation}
which stabilize the maximally entangled state. Therefore, using that state and those measurements we obtain:
\begin{equation}
    \langle \mathcal{F}_d \rangle = d\sum_{m} |\lambda_{0,m-1}|^2  = d
\end{equation}
where we used the fact that $\sum_k |\lambda_{0,k}|^2 = 1$.

To show that this strategy is optimal, we note that this new family of Bell inequalities can be easily written in terms of a Sum-Of-Square (SOS) decomposition:
\begin{equation}
    \beta_d \mathbbm{1} - \mathcal{F}_d \equiv \frac{1}{2d}\sum_\ell P_\ell ^\dagger P_\ell ,
    \label{eq:SOS}
\end{equation}
where:
\begin{equation}
    P_\ell = d\lambda_{0,\ell-1} \mathbbm{1} - A_\ell\otimes  \sum_y \omega^{y\ell} B_y
    \label{eq:P}
\end{equation}
with $\ell=0,\dots,d-1$.
\begin{Theorem}
    The Tsirelson bound of $\langle \mathcal{F}_d \rangle$ is given by $\beta_d = d$ and is attained by measuring $A_j = \left( X^j Z^{j-1} \right)^* $ and $B_k$ given by Eq. \eqref{eq:optbob}  on $\ket{\psi_d} = (1/\sqrt{d}) \sum_i \ket{ii}$.
    \label{prop:newineq}
\end{Theorem}

\begin{proof}
We know that $\langle \mathcal{F}_d \rangle  \geq d$. That $\langle \mathcal{F}_d \rangle  \leq d$ follows from the SOS decomposition and $\sum_k |\lambda_{0,k}|^2 = 1$.
\end{proof}

As one could expect, this inequality is also invariant under symmetry transformations analogous to those presented in the previous sections: $\tau:A_x\mapsto \omega^xA_x,B_y\mapsto B_{y+1}$ leaves the functional invariant and this is also the case for $\tau':A_x\mapsto \omega^{d-1}A_x, B_y\mapsto \omega B_y$. When adding an extra measurement on Bob's side and defining $\langle \bar{\mathcal{F}}_d\rangle  \equiv \langle \mathcal{F}_d\rangle + \langle A_0\otimes B_d\rangle$, the invariance of those transformations implies, following~\cite{dhara2013}, that the maximal quantum violation of the inequality attains optimal global randomness certification for any $d$ for settings $x=1$ and $y=d$, if this maximal violation happens to be unique. This is yet unproven but seems plausible, given the form of the obtained quantum realization and SOS decomposition. However, the fact that this family requires roughly twice more measurements than the first, makes it less appealing from an experimental viewpoint.

The appeal of this new family is that, as it admits an easy SOS decomposition, it could be used to derive some analytical certification statements, as for example self-testing of $X$ and $Z$ measurements in higher dimensions or even larger sets of mutually unbiased bases (MUBs) (note that, whenever $d$ is prime, our optimal strategy uses $d$ MUBs on each side). In fact, for $d=3$, we can give a full self-testing statement for the optimal strategy:
\begin{Theorem}\label{thm:selftesting}
The maximal violation of $\langle \mathcal{F}_3 \rangle$ is attained uniquely (up to local isometries and transposition map) by measurements $A_j = (
X^jZ^{j-1})^*$ for $j=0,1,2$ and $B_k$ given by \eqref{eq:optbob} performed on
$\ket{\psi} = (\ket{00}+\ket{11}+\ket{22})/\sqrt3.$
\end{Theorem}
\begin{proof}
    The proof relies on the results of \cite{Kaniewski_2019} and can be found in Appendix \ref{app:selftesting}. 
\end{proof}

A direct corollary of the self-testing result is that the maximal violation of $\langle \mathcal{F}_3\rangle$ certifies optimal local randomness and the maximal violation of $\langle \bar{\mathcal{F}}_3\rangle  = \langle \mathcal{F}_3\rangle + \langle A_0\otimes B_d\rangle$  certifies optimal global randomness (two random trits).

\section{Conclusions and outlook}\label{sec:conc}

In this work, we provide two families of Bell inequalities that allow to certify optimal global randomness using maximally entangled states of arbitrary local dimension. To obtain the first family, defined in the $(2,d+1;d,d)$ scenario, we exploited an approach developed in Ref.~\cite{dhara2013} which allows to conclude optimal randomness certification from symmetries of the Bell expression and the uniqueness of its maximal violation. That approach, initially presented for Bell scenarios with binary measurements, is extended here to the case of an arbitrary number of outcomes. We then show numerically that this family enables the certification of 
$2\log_2 d$ bits of randomness for $d=3,4$. For $d=3$, our inequality proves to be more resistant to noise than previously known examples when it comes to certify randomness beyond the two-bit barrier defined by the projective qubit protocols \cite{Barcelona2}. This, together with the fact that it requires the minimal number of measurements on one of the parties, make this construction potentially appealing from an experimental point of view. Finally, it is worth noting that the maximal quantum value of this functional is attained by a strategy in which the parties share a maximally entangled state of two qudits and one of them performs generalized $X$ and $Z$ Pauli measurements.

Based on the above construction, we then proposed a second family of Bell inequalities in the $(d,d;d,d)$ scenario, for which an optimal quantum violation and its realization can be analytically obtained. This realization requires the maximally entangled state and generalized $X$ and $Z$ among the measurements of one of the parties. Moreover, for the $d=3$ case, we proved that the maximal violation of this inequality self-tests, similarly to Ref. \cite{Kaniewski_2019}, the maximally entangled state and three mutually unbiased bases. This automatically proves the DI certification of optimal local randomness and, extending the scenario by adding a measurement on one of the parties, it allows to certify optimal global randomness from projective measurements.

Several questions remain open. First, an analytical proof that the new inequality certifies maximal global randomness for projective measurements for arbitrary $d$ remains to be found. It would also be interesting to establish self-testing statements for the introduced inequalities, in particular for the one admitting an SOS decomposition at the first level of the NPA hierarchy. On the one hand, this would provide a way to prove that the inequality certifies the maximal amount of $2\log_2 d$ bits for any $d$; on the other hand, it would enable certification the $X$ and $Z$ observables in arbitrary dimension, solving an open problem in the field. Third, it remains open whether the new inequality is optimal in terms of noise robustness or whether better inequalities exist in this respect. Finally, possible generalizations to the multipartite case also remain to be explored.

\section{Acknowledgments}\label{sec:ack}
This work was supported by the European Union’s Horizon Europe research and innovation programme under grant agreement No 101080086/NeQST,  the Government of Spain (Severo Ochoa CEX2019-000910-S and FUNQIP), Fundació Cellex, Fundació Mir-Puig, Generalitat de Catalunya (CERCA program), the EU (QSNP, 101114043, Quantera COMPUTE No 101017733) and the AXA Chair in Quantum Information Science. RD acknowledges funding from the European Union’s Horizon Europe research and innovation programme under the Marie Skłodowska-Curie grant agreement No. 101081441. The Center for Quantum Enabled Computing project is carried out within the International Research Agendas programme of the Foundation for Polish Science co-financed by the European Union under the European Funds for Smart Economy 2021-2027 (FENG). 

\section{Note added}

Notice that similar results were
independently developed in \cite{Mate} where a device-independent scheme for maximal global randomness certification from projective measurements was provided.

\bibliography{references}

\newpage
\appendix
\onecolumngrid
\newpage

\section{Relevant bounds for $\mathcal{I}_d$}
\label{app:ibounds}

\subsection{Optimal classical and non-signaling values for $\mathcal{I}_d$}
\begin{prop}
The non-signaling bound of $\mathcal{I}_d$ is:
\begin{equation}
        \beta_\mathcal{NS}(d) = 2d ,
\end{equation}
and the classical bound is given by:
\begin{equation}
           \beta_\mathcal{L}(d) = 2\left[ 1 + 2 \sum_{\ell=1}^{\lfloor \frac{d+1}{2}\rfloor - 1} \cos\left( \frac{\pi \ell }{d} \right) f_\ell \right] , 
\end{equation}
where $ f_\ell = \cos (\pi/d)$ if $\ell$ is odd and $f_\ell = 1$ otherwise. 
\label{prop:LandNSbounds}
\end{prop}
\begin{proof}
The Bell operator can be written as:
\begin{equation}
    \mathcal{I}_d \equiv \frac{1}{2}\sum_{y=0}^{d-1}  \left( A_0 + \omega^y A_1 \right) \otimes B_y^\dagger + \text{h.c.,}
\end{equation}
and the local deterministic strategies correspond to set $\langle A_x \rangle = \omega^{a_x}$, $\langle B_y \rangle = \omega^{b_y}$ and $\langle A_x B_y \rangle = \omega^{a_x+b_y}$, where $a_x,b_y \in \{0,\dots,d-1\}$ are the deterministic outputs that each party observe when they choose inputs $x$ and $y$ respectively. The mean value of the Bell operator, evaluated in a local deterministic strategy is then:
\begin{equation}
\begin{split}
    \langle \mathcal{I}_d\rangle &= \frac{1}{2}\sum_{y=0}^{d-1}\left( 1 + \omega^{y+a_1-a_0}\right) \omega^{a_0+b_y} + c.c. \\  &= 2\sum_{y=0}^{d-1}\cos \left[\frac{\pi}{d}(y+a_1-a_0)\right] \cos\left[\frac{\pi}{d}(y+a_1-a_0)+\frac{2\pi}{d}(a_0+b_y)\right] .
\end{split}
\end{equation}
From the last expression it can be seen that one can fix $a_0 = a_1 = 0$ without losing generality, as different values of $a_x$ can be absorbed into Bob's variables by means of the relabeling $b_y \rightarrow b_y+a_0$ and changing the index of the sum. So, the problem reduces to find the set $\{b_y\}_{y=0,\dots,d-1}$ that maximizes:
\begin{equation}
    \langle \mathcal{I}_d\rangle = 2\sum_{y=0}^{d-1}\cos \left(\frac{\pi y}{d}\right) \cos\left[\frac{\pi}{d}\left(y+2b_y\right)\right] .
\end{equation}
Now, as each term on the previous sum depends on a single $b_y$, the problem reduces to picking them in such a way that each term is maximized independently: if $\cos \left(\frac{\pi y}{d}\right) >0$ and $y$ is even, we can pick $b_y = -y/2 \mod d$ which gives $\cos\left[\frac{\pi}{d}\left(y+2b_y\right)\right] = 1$. On the other hand, if $\cos \left(\frac{\pi y}{d}\right) >0$ and $y$ is odd, we can pick $b_y = (1-y)/2\mod d $ which gives $\cos\left[\frac{\pi}{d}\left(y+2b_y\right)\right] = \cos\left(\pi / d\right)$. The terms with $\cos \left(\frac{\pi y}{d}\right) <0$ are optimized in a similar fashion. Note also that, except for the first term (which is equal to $1$), all remaining terms come in pairs, each of them corresponding to each class $\cos \left(\frac{\pi y}{d}\right) >0$ and $\cos \left(\frac{\pi y}{d}\right) <0$ (in the even $d$ case, there is a term with $\cos \left(\frac{\pi y}{d}\right)  =0$ which is irrelevant). Everything put together, we arrive to the desired formula:
\begin{equation}
    \max_{p\in\mathcal{L}} \, \langle \mathcal{I}_d\rangle = 2\left[ 1 + 2 \sum_{\ell=1}^{\lfloor \frac{d+1}{2}\rfloor - 1} \cos\left( \frac{\pi \ell }{d} \right) f_\ell \right] , 
\end{equation}
where $f_\ell = 1$ if $\ell$ is even and $f_\ell = \cos(\pi/d)$ otherwise. 

For the non-signaling bound, we just write:
\begin{equation}
    \mathcal{I}_d \equiv \frac{1}{2}\sum_{y=0}^{d-1}  \left( A_0B_y + \omega^y A_1B_y \right) + \text{h.c.,}
    \label{eq:iforns}
\end{equation}
and note that nonlocal vertices of the non-signaling polytope correspond to $\langle A_x \rangle = 0$, $\langle B_y \rangle = 0$ and $\langle A_x B_y \rangle = \omega^{g_{xy}}$, which are PR-like non-local deterministic strategies such that $a+b = g_{xy}  \mod d$. This means that with non-signaling resources we can make the expectation value of each two-body correlator in \eqref{eq:iforns} equal to one and therefore $\max_{p\in\mathcal{NS}}  \,  \langle \mathcal{I}_d\rangle = 2d$. 

\end{proof}


\subsection{Upper bound for the optimal quantum value of $\mathcal{I}_d$}
\begin{prop}
The quantum bound of $\mathcal{I}_d$ satisfies:
\begin{equation}
        \beta_\mathcal{Q}(d) \leq d\sqrt{2} .
\label{eq:qbound}
\end{equation}
This upper bound is, in general, not tight.
\label{prop:qbound}
\end{prop}
\begin{proof}
We define the operators:
\begin{equation}
    P_\ell \equiv \mathbbm{1} - \frac{A_0+\omega^\ell A_1} {\sqrt{2}} \otimes B_\ell ,
\end{equation}
which satisfy:
\begin{equation}
    P_\ell^\dagger P_\ell = \mathbbm{1} - \frac{A_0+\omega^\ell A_1} {\sqrt{2}}\otimes B_\ell +\frac{\omega^{-\ell}A_0 ^\dagger A_1}{2}\otimes\mathbbm{1} + \text{h.c.} ,
\end{equation}
and hence:
\begin{equation}
    \sum_{\ell=0}^{d-1} P_\ell^\dagger P_\ell = d\, \mathbbm{1} - \frac{1}{\sqrt2} \sum_{\ell=0}^{d-1}\left(A_0+\omega^\ell A_1\right)\otimes B_\ell + \frac{A_0 ^\dagger A_1}{2}\otimes \sum_{\ell=0}^{d-1} \omega^{-\ell}\mathbbm{1} + \text{h.c.} =  2d\, \mathbbm{1} - \sqrt2 \,\mathcal{I}_d ,
    \label{eq:sosId}
\end{equation}
where, in the last identity, we used that $\sum_{\ell=0}^{d-1} \omega^{-\ell}=0$. Being a sum of squares, we know that $\sum_{\ell=0}^{d-1} P_\ell^\dagger P_\ell$ is a semi-definite positive operator, so we conclude that for any quantum strategy:
\begin{equation}
    \langle \mathcal{I}_d \rangle \leq \sqrt2 d ,
    \label{eq:upperboundId}
\end{equation}
which is the desired bound.

\end{proof}
This upper bound is only tight for $d=2$, and the level at which the NPA relaxation becomes tight increases with the dimension (see Appendix \ref{app:numericalbounds} for details).

\subsection{Numerical bounds for $\mathcal{I}_d$}
\label{app:numericalbounds}

The following table summarizes, for $d\leq 6$, numerical upper bounds (found by means of NPA relaxations) and lower bounds (by optimizing over explicit quantum strategies using a see-saw algorithm) for $\mathcal{I}_d$ and compares them to the analytical upper bound \eqref{eq:qbound} and the conjectured (and also analytical) optimal quantum value \eqref{eq:tsir_bound}. Note that the upper bound \eqref{eq:qbound} is tight for level $1$ of NPA in all cases, which is expected as the SOS decomposition used to prove it only involves level $1$ terms. Also note that the only dimension in which this upper bound matches the optimal quantum value is $d=2$, for which the inequality reduces to CHSH.

\begin{center}
\begin{table}[h]
\begin{tabular}{|c|c|c|c|c|c|c|}
\hline
\textbf{\,\,\,} & \makecell{\textbf{Upper} \\ \textbf{bound \eqref{eq:qbound}}}  & \makecell{\textbf{Conjectured optimal} \\ \textbf{quantum value \eqref{eq:tsir_bound}}}  & \makecell{\textbf{NPA} \\ \textbf{(level 1)}} &\makecell{\textbf{NPA } \\ \textbf{(level 1+AB)}} & \makecell{\textbf{NPA} \\ \textbf{(level 2)}} & \makecell{\textbf{Lower bound} \\ \textbf{(see-saw)}}  \\
\hline
$\mathcal{I}_2$ & $2\sqrt2 = 2.828\dots$   & $2 \left[ \sin \left(\frac{\pi}{4} \right)  \right] ^{-1} = 2.828\dots$ & $2.828\dots$   & - & - &  $2.828\dots$    \\
$\mathcal{I}_3$ & $3\sqrt2 = 4.242\dots$  & $2 \left[ \sin \left(\frac{\pi}{6} \right)  \right] ^{-1} = 4 \quad \quad\quad \; $  & $4.242\dots$ & $4.000\dots$  & - &  $4.000\dots$  \\
$\mathcal{I}_4$ & $4\sqrt2 = 5.656\dots$ &$2 \left[ \sin \left(\frac{\pi}{8} \right)  \right] ^{-1} = 5.226\dots$  & $5.656\dots$ & $5.226\dots$  & - &  $5.226\dots$   \\
$\mathcal{I}_5$ & $5\sqrt2 = 7.071\dots$ & $2 \left[ \sin \left(\frac{\pi}{10} \right)  \right] ^{-1} =6.472\dots $ & $7.071\dots$ & $6.478\dots$ & $6.472\dots$ &  $6.472\dots$  \\
$\mathcal{I}_6$ & $6\sqrt2 = 8.485\dots$ & $2 \left[ \sin \left(\frac{\pi}{12} \right)  \right] ^{-1} =7.727\dots $ & $8.485\dots$ & $7.745\dots$ & $7.727\dots$ &  $7.727\dots$  \\
\hline
\end{tabular}
\caption{Summary of analytical and numerical bounds for $\mathcal{I}_d$. All numerical results (last five columns) are truncated to ease readability, but matching values differ from their analytical counterparts by less than $10^{-6}$ in all cases.}
\label{tab:num_res}
\end{table}
\end{center}

\subsection{Ratio between local and quantum optimal values for $\mathcal{I}_d$}
\label{app:ratio}

\begin{figure}
    \centering
    \includegraphics[width=0.5\linewidth]{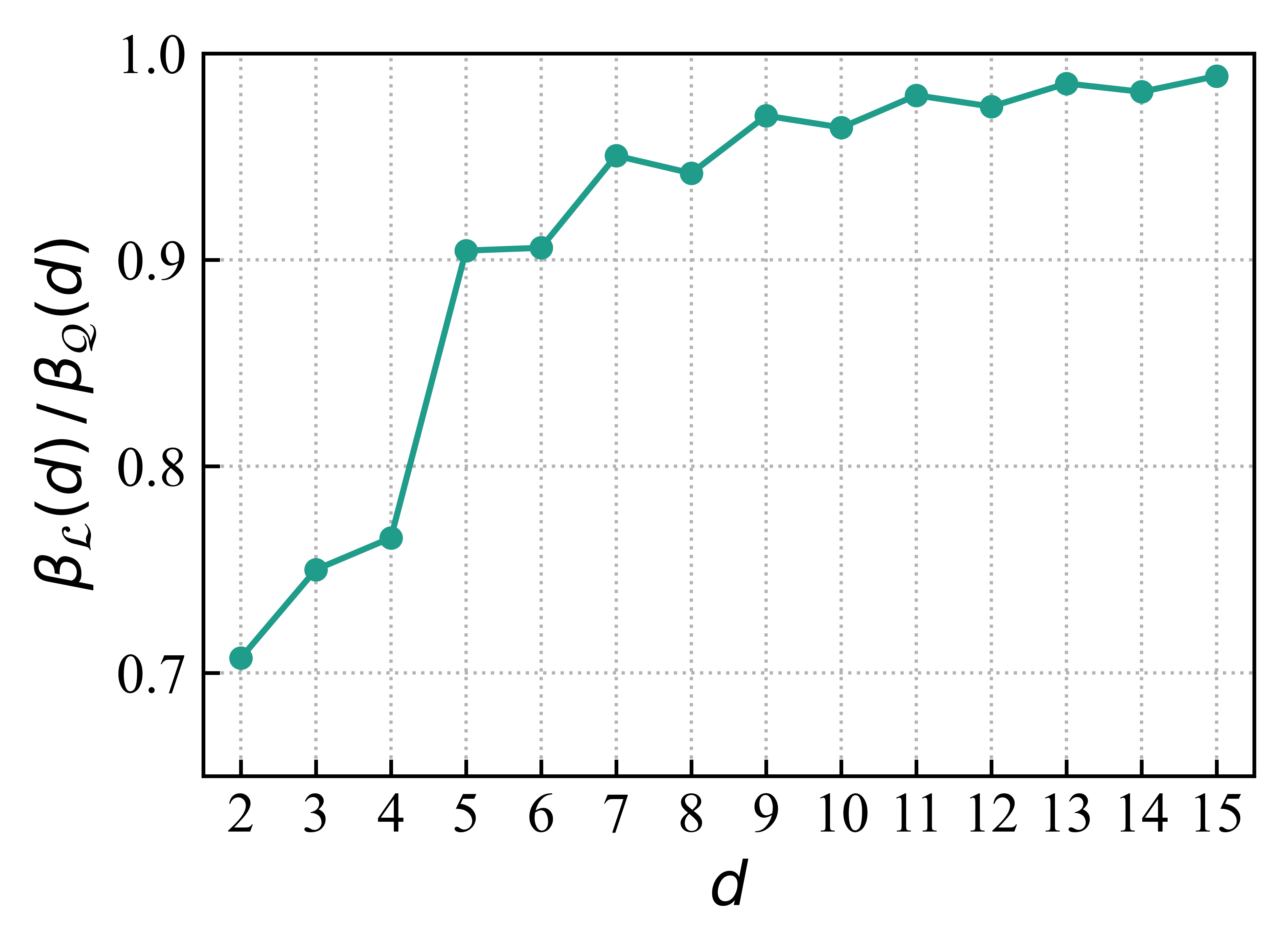}
    \caption{Ratio between the local and quantum bounds of $\mathcal{I}_d$ as a function of the dimension $d$.}
    \label{fig:lboundoverqbound}
\end{figure}

Figure \ref{fig:lboundoverqbound} shows the ratio $\beta_\mathcal{L}(d)/\beta_\mathcal{Q}(d)$ for $d\leq 15$. For low dimensions, our construction exhibits a wide range of nonlocal values: it allows to certify nonlocality in presence of up to $5 \%$ of depolarizing noise up to $d=9$. In particular, for $d=3$, nonlocality can be witnessed even against $25 \%$ of depolarizing noise.

\newpage

\section{Device-independent guessing probability as a semi-definite program and numerical certification.}
\label{app:randsdp}

Consider a semi-definite program (SDP) in the standard canonical form
\begin{align}
	\max_{X} \quad & \langle C, X \rangle \\
	\text{s.t.} \quad & \langle A_i, X \rangle = b_i \quad \text{for all } i, \\
	& X \succeq 0,
\end{align}
where $X$ is a symmetric positive semi-definite matrix, $\langle A, B \rangle = \mathrm{Tr}(A^\top B)$ denotes the matrix inner product, and all constraints are linear in $X$~\cite{mironowicz2024semi}.

We use the Navascu\'es--Pironio--Ac\'in (NPA) hierarchy to approximate the set of quantum correlations. One introduces a moment matrix $\Gamma$ indexed by words in measurement operators, with entries given by expectations of operator products, $\Gamma_{u,v} = \langle \psi \mid u^\dagger v \mid \psi \rangle$. The matrix satisfies $\Gamma \succeq 0$, and linear constraints encode normalization, orthogonality of projectors, and commutation relations between Alice and Bob.

The hierarchy is defined by restricting the set of words. At level~1 one includes words of length at most one, while level~$1+\mathrm{AB}$ additionally includes products of one Alice and one Bob operator. Higher levels include longer words, such as $\mathrm{AAB}$ or $\mathrm{ABB}$, and provide tighter approximations to the quantum set.

Given a Bell functional $F$, we express it as a linear combination of probabilities $P(a,b\mid x,y)$, each represented as a linear function of entries of $\Gamma$. We then solve the SDP maximizing $F(\Gamma)$ subject to $\Gamma$ belonging to the chosen NPA level. The optimal value gives an upper bound on the quantum value, i.e.\ the Tsirelson bound.

From the optimal moment matrix $\Gamma^\ast$ one reconstructs the corresponding distribution $P^\ast(a,b\mid x,y)$ by reading the appropriate entries of $\Gamma^\ast$.

To compute the guessing probability we use the method of Nieto-Silleras based on decomposing the observed distribution into a convex mixture of strategies corresponding to the possible outcomes guessed by the eavesdropper~\cite{nieto2014using,bancal2014more}. In our notation this amounts to introducing a family of moment matrices $\Gamma^{\alpha\beta}$, each associated with the strategy in which Eve guesses the pair of outcomes $(\alpha,\beta)$. Each matrix $\Gamma^{\alpha\beta}$ satisfies the NPA constraints at the chosen relaxation level, in particular $\Gamma^{\alpha\beta}\succeq0$ together with all linear constraints imposed by the operator relations. In addition, we impose the reconstruction condition
\begin{equation}
    \label{eq:GP_reconstruction}
    \sum_{\alpha,\beta}\Gamma^{\alpha\beta} = \Gamma,
\end{equation}
which at the level of probabilities corresponds to
\begin{equation}
    \label{eq:GP_prob_reconstruction}
    \sum_{\alpha,\beta} P^{\alpha\beta}(a,b|x,y) = P(a,b|x,y).
\end{equation}

For fixed inputs $x,y$, the global guessing probability is defined as
\begin{equation}
    \label{eq:global_guessing_probability}
    G(AB|x,y) = \max_{\{\Gamma^{\alpha\beta}\}} \sum_{\alpha,\beta} P^{\alpha\beta}(\alpha,\beta|x,y), 
\end{equation}
where the maximization is taken over all families of moment matrices $\{\Gamma^{\alpha\beta}\}$ satisfying the NPA constraints and the reconstruction condition~\eqref{eq:GP_reconstruction}. In addition, the reconstructed distribution is required to coincide with a target distribution or to satisfy a given Bell constraint, for example $F(\Gamma)=\beta$ or $F(\Gamma)\ge\beta-\varepsilon$. Since all probabilities
$P^{\alpha\beta}(a,b|x,y)$
are linear functions of the entries of the corresponding moment matrices $\Gamma^{\alpha\beta}$, the resulting optimization problem remains a semi-definite program in standard form. Solving this SDP yields the optimal quantum upper bound on the global guessing probability compatible with the observed Bell violation.

\subsection{Numerical proof of optimal randomness certification using $\tilde{\mathcal{I}}_d$ for $d=3$ and $d=4$}
\label{app:optrand34}

The following table summarizes numerical bounds for the device-independent guessing probability of both outputs when the full distribution $p(ab|xy)$ is fixed by the optimal quantum strategy defined in Conjecture \ref{conj:i}.
\begin{center}
\begin{table}[h]
\begin{tabular}{|c|c|c|c|c|}
\hline
$d$ & Target $p_{guess}$ & Target randomness & Upper bound for $p_{guess}$ & Lower bound for randomness \\
\hline
$3$ & $1/9 = 0.111\dots$ & $\log_2(9) \simeq 3.169$ & $0.11112$ & $3.169$ \\
\hline
$4$ & $1/16 = 0.0625$ & $\log_2(16) = 4$ & $0.06251$ & $3.999$\\
\hline
$5$ & $1/25 = 0.04$ & $\log_2(25) \simeq 4.644$ &  $0.04015$ & $4.638$\\
\hline

\end{tabular}
\caption{Summary of numerical bounds for the guessing probability and corresponding certifiable randomness. All values are computed using level $2$ of the NPA hierarchy.}
\label{tab:rand_num}
\end{table}
\end{center}

For $d=5$, computation took more than $26$ days, requiring up to 7.21 TiB of physical memory, so higher levels are not possible to run in our available hardware. It can also be seen that the level of NPA required to attain optimality is higher than the level at which the inequality becomes tight (for example, for $d=4$ the inequality is tight for level $1+AB$, although one needs to go to level $2$ to reproduce the results shown in Table \ref{tab:rand_num}). For those two reasons, optimality cannot be shown for $d=5$ up to numerical precision, but the partial results still indicate that our construction could certify optimal randomness in $d>4$.

\subsection{Randomness robustness for $\tilde{\mathcal{I}}_3$ and $\tilde{\mathcal{F}}_3$}
\label{app:rob}

We analyze the robustness to noise of the inequality $\tilde{\mathcal{I}}_3$ defined in Eq.~\eqref{eq:i3bar}, as well as $\tilde{\mathcal{F}}_3$. The latter is obtained from the family $\tilde{\mathcal{F}}_d$, is defined as a modification of $\mathcal{F}_d$ (Eq.~\eqref{eq:bell_general_expr}) in the same spirit as for $\tilde{\mathcal{I}}_3$: namely, we include an additional measurement setting for Bob and introduce the extra correlator $\langle A_0 B_d^\dagger \rangle$. The resulting Bell expression is
\begin{equation}
    \tilde{\mathcal{F}}_d = \frac{1}{2}\left[\sum_{x,y} \lambda_{0,x-1}^* \,\omega^{xy}\, \langle A_x B_y^\dagger \rangle + \langle A_0 B_d^\dagger \rangle \right] + \mathrm{h.c.}\, .
    \label{app:bell_general_expr}
\end{equation}

To assess noise robustness, we consider a Werner-type mixture. At the level of probability distributions, this corresponds to
\begin{equation}
	\label{eq:noise_distr}
	p_v = (1-v)\, p_{\text{Tsir}} + v\, p_{\text{unif}},
\end{equation}
where $p_{\text{Tsir}}$ denotes the distribution achieving the maximal quantum violation of the Bell inequality, $p_{\text{unif}}$ is the uniform distribution associated with the maximally mixed state, and $v \in [0,1]$ is the noise parameter.

In Fig.~\ref{fig:robustness}, we plot the certified randomness as a function of the noise parameter $v$, using the NPA hierarchy at level 2. We observe that the robustness to noise of the two inequalities is very similar.

\begin{figure}[h!]
    \centering
    \includegraphics[width=0.5\linewidth]{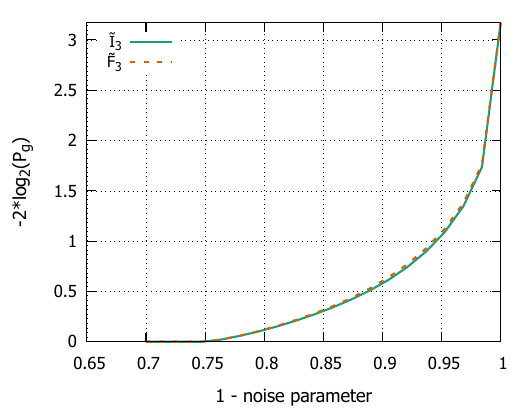}
    \caption{Certified randomness as a function of the noise parameter $v$ for $\tilde{\mathcal{I}}_3$ and $\tilde{\mathcal{F}}_3$ under the Werner-type noise model of Eq.~\eqref{eq:noise_distr}.}
    \label{fig:robustness}
\end{figure}

In~\cite{Barcelona2}, we perform a more systematic analysis of robustness, in which $\tilde{\mathcal{I}}_3$ is compared with other known Bell inequalities with the same number of outcomes. We find that, among inequalities that certify maximal randomness in the ideal (noise-free) scenario, $\tilde{\mathcal{I}}_3$ exhibits improved robustness to noise.

\section{Optimal quantum strategy for $\mathcal{I}_d$}
\label{app:optstrat}

\label{app:optimalconstruction}
Here we show that $B_y$ defined in \eqref{eq:optbob} with coefficients \eqref{eq:optcoeffs} are proper observables, i.e. that they are unitary and with eigenvalues equal to $d$-th roots of unity. To this end we write $B_y$ as a product of $X$ and a polynomial of $XZ$ and focus on the spectrum of the latter one.
For clarity we recall the form of observables
\begin{equation}
    B_y = \sum_{k=0}^{d-1} \lambda_{y,k} X^{k+1} Z^k ,
\end{equation}
where the full coefficients are
\begin{equation}
 \lambda_{y,k} = \frac{(-1)^k\omega^{\frac{k(k+1)}{2}}\omega^{-y(1+k)}}{d\sin\left(\frac{\pi}{d} (k+\frac{1}2{}) \right)} .
    \label{eq:optcoeffs}
\end{equation}

Using $ZX=\omega XZ$ can write
\begin{equation}
B_y = X \sum_{k=0}^{d-1} \lambda_{y,k}\omega^{-\frac{k(k-1)}{2}}(XZ)^k.
\end{equation}
We see that $B_y=XW_y$, where $W_y$ is defined as
\begin{equation}
\label{polynomial_xz}
    W_y=\sum_{k=0}^{d-1} \lambda_{y,k}\,\omega^{-\frac{k(k-1)}{2}}(XZ)^k,
\end{equation}which is a polynomial of $XZ$. To prove properties of $B_y$, we look at the spectrum of $W_y$ and at the action of $X$ on the eigenvectors $\{|u_j\rangle\}_{j\in\mathbb{Z}^d}$ of $XZ$. Note that
\begin{equation}
    (XZ)^d=\omega^{d(d-1)/2}\mathbbm{1}=(-1)^{d-1}\mathbbm{1},
\end{equation}
which means that $XZ$ has eigenvalues $\omega^j$ for odd $d$ and $\omega^{j+\frac{1}{2}}$ for even $d$ ($j\in\mathbb{Z}^d$). Using $ZX=\omega XZ$ we observe that
\begin{equation}
    (XZ)X|u_j\rangle=\omega X(XZ|u_j\rangle)=
    \begin{cases}
        \omega^{j+1}X|u_j\rangle,&d \text{ is odd,}\\
        \omega^{j+\frac{3}{2}}X|u_j\rangle,&d\text{ is even.}
    \end{cases}
\end{equation}
Therefore, regardless of dimension, $X|u_j\rangle$ is an eigenvector of unitary operator $XZ$ and it is exactly $|u_{j+1}\rangle$ up to a phase factor. We have
\begin{equation}
    X|u_j\rangle=e^{i\xi_j}|u_{j+1}\rangle.
\end{equation}
Property $X^d=\mathbbm{1}$ implies that phase factors $e^{i\xi_j}$ multiply to one as
\begin{equation}
    |u_j\rangle =X^d|u_j\rangle=\left(\prod_{i=j}^{j+d-1} e^{i\xi_i}\right)|u_{j}\rangle.
\end{equation}
Let $W_y|u_j\rangle=w_{y,j}|u_j\rangle$, which we can write based on \eqref{polynomial_xz}. Then
\begin{equation}
    B_y |u_j\rangle = XW_y|u_j\rangle=w_{y,j}X|u_j\rangle=w_{y,j}e^{i\xi_j}|u_{j+1}\rangle.
\end{equation}
Repeating this $d$ times gives
\begin{equation}
\left(B_y\right)^d |u_j\rangle= \left(\prod_{i=j}^{j+d-1}w_{y,i}\right)\left(\prod_{i=j}^{j+d-1} e^{i\xi_i}\right)|u_j\rangle=\left(\prod_{i=0}^{d-1}w_{y,i}\right)|u_j\rangle,
\label{eq:eigB}
\end{equation}
where we use $\prod_{i=0}^{d-1} e^{i\xi_i}=1$. Shifting summation bounds is irrelevant. To calculate the expression above we have to find eigenvalues $\{w_{y,i}\}_{i\in \mathbb{Z}^d}$ of $W_y$, which are different depending on the parity of $d$.

\begin{itemize}
    \item \textbf{Odd $d$ case} 
\end{itemize}
Eigenvalues of $XZ$ are $\omega^j$, $j\in\mathbb{Z}^d$. Therefore using \eqref{eq:optcoeffs} we explicitly write the spectrum of $W_y$ as
\begin{equation}
\sigma(W_y) = \left\{w_{y,j}=\sum_{k=0}^{d-1} \frac{(-1)^k \omega^{-y(1+k)}}{d \sin \frac{\pi}{d} (k+\frac{1}{2})} \omega^{k}\omega^{jk}, j\in\{0, 1, \dots, d-1\} \right\}.
\label{spectrum_p_odd}
\end{equation}
Recall the Dirichlet kernel \begin{equation}
    D_M(x) = \sum_{m=-M}^{M} e^{imx} = \frac{\sin\left(\left(M+\frac{1}{2}\right)x\right)}{\sin(x/2)}
\end{equation}
Put $M = \frac{d-1}{2}$ and $x= \frac{\pi}{d}(2k+1)$
\begin{equation}
        D_{\frac{d-1}{2}}\left(\frac{\pi}{d}(2k+1)\right) = \frac{\sin\left(\frac{\pi}{2}(2k+1)\right)}{\sin\left(\frac{\pi}{2d}(2k+1)\right)} = \frac{(-1)^k}{\sin\left(\frac{\pi}{2d}(2k+1)\right)}= \sum_{m=-\frac{d-1}{2}}^{\frac{d-1}{2}} e^{im \left(\frac{\pi}{d}(2k+1)\right)}.
\end{equation}
We place into the formula from \eqref{spectrum_p_odd} 
\begin{equation}
   w_{y,j}= \sum_{k=0}^{d-1} \frac{1}{d} \omega^{-y(1+k)} \omega^{k(j+1)} \sum_{m=-\frac{d-1}{2}}^{\frac{d-1}{2}} e^{im \left(\frac{\pi}{d}(2k+1)\right)}.\end{equation}
   We now write $e^{im \frac{\pi}{d}(2k+1)} = \omega^{mk} e^{im \frac{\pi}{d}}$ and change the order of summation.
\begin{equation}
        w_{y,j}= \frac{\omega^{-y}}{d} \sum_{m=-\frac{d-1}{2}}^{\frac{d-1}{2}} e^{im \frac{\pi}{d}} \sum_{k=0}^{d-1} \omega^{k(j+1-y+m)}
\end{equation}
Note that, while the last term vanishes unless $m\equiv y-j-1 \pmod{d}$, $y-j-1$ may not be in the domain of $m$. Therefore
 \begin{equation}
 \label{symmetric_modulo}
     w_{y,j}= \omega^{-y}e^{i\frac{\pi}{d}(y-j-1\text{ smod } d)},
 \end{equation}
 where we introduce symmetrized modulo addition
 \begin{equation}
     a\text{ smod }d\coloneq (a\text{ mod }d)-\frac{d-1}{2},
 \end{equation}
 which is a standard modulo with results shifted by $-(d-1)/2$. Now we calculate the product of eigenvalues
 \begin{equation}
    \label{polynomial_eigen_odd}
    \prod_{k=0}^{d-1}w_{y,k}=\prod_{k=0}^{d-1}\omega^ye^{i\frac{\pi}{d}(y-k-1\text{ smod }{d})}=(\omega^d)^y\text{exp}\left(i\frac{\pi}{d}\sum_{k=0}^{d-1}\left[(y-k-1) \text{ smod } d\right]\right).
\end{equation}
Note that $\sum_{k=0}^{d-1}\left[(y-k-1) \text{ smod } d\right]$ is a sum of integers from $-\frac{d-1}{2}$ to $\frac{d-1}{2}$, which is simply zero. Also $\omega^d=1$, which gets us to
\begin{equation}
    \prod_{k=0}^{d-1}w_{y,k}=1,
\end{equation}
which by virtue of \eqref{eq:eigB} means that $\left(B_y\right)^d =\mathbbm{1}$.

\begin{itemize}
    \item \textbf{{Even $d$ case}}
\end{itemize} 
Now eigenvalues of an operator $XZ$ are $\omega^{j+\frac{1}{2}}$. The spectrum
\begin{equation}
\sigma(W_y) = \left\{w_{y,j}=\sum_{k=0}^{d-1} \frac{(-1)^k \omega^{-y(1+k)}}{d \sin \frac{\pi}{d} (k+\frac{1}{2})} \omega^{(j+\frac{3}{2})k}, j\in\{0, 1, \dots, d-1\} \right\}.
\end{equation}
We use the same Dirichlet kernel as above, while shifting summation bounds as $\frac{d-1}{2}$ is not an integer.
\begin{equation}
        D_{\frac{d-1}{2}}\left(\frac{\pi}{d}(2k+1)\right) = \frac{(-1)^k}{\sin\left(\frac{\pi}{2d}(2k+1)\right)}= \sum_{m=-\frac{d-1}{2}}^{\frac{d-1}{2}} e^{im \left(\frac{\pi}{d}(2k+1)\right)}=\sum_{m=-\frac{d}{2}}^{\frac{d}{2}-1} e^{i(m+\frac{1}{2}) \left(\frac{\pi}{d}(2k+1)\right)}.
\end{equation}
Now we replace the kernel into the formula
\begin{equation}
   w_{y,j}= \sum_{k=0}^{d-1} \frac{1}{d} \omega^{-y(1+k)} \omega^{k(j+\frac{3}{2})} \sum_{m=-\frac{d}{2}}^{\frac{d}{2}-1} e^{i\frac{\pi}{d}(m+\frac{1}{2}) \left(2k+1\right)}.\end{equation}
Using $e^{i\frac{\pi}{d}(m+\frac{1}{2})(2k+1)} = \omega^{(m+\frac{1}{2})k} e^{i \frac{\pi}{d}(m+\frac{1}{2})}$ we get
\begin{equation}
\label{m_domain}
        w_{y,j}= \frac{\omega^{-y}}{d} \sum_{m=-\frac{d}{2}}^{\frac{d}{2}-1} e^{i \frac{\pi}{d}(m+\frac{1}{2})} \sum_{k=0}^{d-1} \omega^{k(j+\frac{3}{2}+\frac{1}{2}-y+m)}.
\end{equation}
The sum over $k$ vanishes unless $m\equiv y-j-2 \pmod{d}$. Here we need modulo shifted by $-\frac{d}{2}$ to match the domain of $m$.
 \begin{equation}
 \label{symmetric_modulo2}
     w_{y,j}= \omega^{-y}e^{i\frac{\pi}{d}\left[(y-j-2\text{ mod } d)-\frac{d}{2}+\frac{1}{2}\right]}=\omega^{-y}e^{i\frac{\pi}{d}(y-j-2\text{ smod } d)}.
 \end{equation}
 Note that this is the same formula as in \eqref{symmetric_modulo} up to transformation $j\xrightarrow{}j+1$. Therefore, just as in \eqref{polynomial_eigen_odd}, we have
 \begin{equation}
    \prod_{k=0}^{d-1}w_{y,k}=1.
\end{equation}
Combined with \eqref{eq:eigB} it shows that $(B_y)^d=\mathbbm{1} $ for even $d$.

Note that for any dimension $W_y$ happens to be a unitary operator as its eigenvalues given in \eqref{symmetric_modulo} and \eqref{symmetric_modulo2} are of magnitude one. Thus we can conclude that $B_y=XW_y$ is also unitary. Therefore $B_y$ is a proper observable. 
%

\section{Self-testing using shifted Bell operator $\mathcal{F}_3$}\label{app:selftesting}
Here we show that the family of Bell inequalities that we introduced in Section \ref{sec:2fam} is capable of self-testing maximally entangled state and a set of optimal observables for local dimension $d=3$ up to transposition map. We begin stating Theorem \ref{thm:selftesting}) in more detailed way.
\begin{Theorem}
Assume that the state $\ket{\psi}\in\mathcal{H}_A\otimes\mathcal{H}_B$ and observables $A_x$
and $B_y$ attain the maximal quantum value of Bell operator $\mathcal{F_3}$ defined in \eqref{eq:bell_general_expr}. Then, there exist local unitary operations $U_A$ and $U_B$ such that $U_A:\mathcal{H}_A\to \mathbb{C}^3\otimes\mathcal{H}_{A''}$ and $U_B:\mathcal{H}_B\to \mathbb{C}^3\otimes\mathcal{H}_{B''}$ and 
\begin{equation}\label{StateST}
    U_A\otimes U_B\ket{\psi}=\ket{\psi_{+}}\otimes |\mathrm{aux}\rangle_{A''B''}
\end{equation}
and
\begin{eqnarray}
U_A\, A_{0}\, U_A^\dagger &=& Z\otimes\mathbbm{1},\\
U_A\, A_{1}\, U_A^\dagger &=& X\otimes P_{1} + X^2 \otimes P_{2},\\
    U_A\, A_{2}\, U_A^\dagger &=& \omega^2 X^{2} Z^{2} \otimes P_{1} + \omega X Z^{2} \otimes P_{2},
\end{eqnarray}
where $P_1+P_2=\mathbbm{1}_{A''}$, and
\begin{eqnarray}
U_B\, B_{y}\, U_B^\dagger &=& B'_y\otimes Q_{1} + (B'_y)^T \otimes Q_{2},
\end{eqnarray}
where $Q_1+Q_2=\mathbbm{1}_{B''}$. $B'_y$ are Bob's optimal observables, written explicitly in \eqref{eq:bobsalld3}.
\end{Theorem}
\begin{proof} Coefficients $\lambda_{y,k}$ given in \eqref{eq:optcoeffs} for $d=3$ give raise to the shifted Bell operator defined in \eqref{eq:bell_general_expr} of the following form
\begin{align}
\mathcal{F}_3&= \frac{1}{6}\bigg[(2A_0 + 2A_1 - A_2) \otimes B_0 \nonumber \\
  &\quad +  (2A_0 + 2\omega A_1 - \omega^2 A_2) \otimes B_1 \nonumber \\
  &\quad +  (2A_0 + 2\omega^2 A_1 - \omega A_2) \otimes B_2\bigg]+\text{h.c.,}
\end{align}
where $A_x$ and $B_y$ are at the moment arbitrary observables of Alice and Bob. 
One checks directly that the following sum-of-squares (SOS) decomposition~\cite{doherty2008quantum}, which is different from Eq.~\eqref{eq:P}, also 
holds true
\begin{equation}
    3\mathbbm{1} - \mathcal{F}_3=\frac{1}{4}\sum_{i=1}^{3}(L_i^{\dagger}L_i+L_iL_i^{\dagger}),
\end{equation}
where 
\begin{equation}
    L_1=\mathbbm{1} - \frac{1}{3}(2A_0 + 2A_1 - A_2) \otimes B_0,\qquad L_2=\mathbbm{1} - \frac{1}{3}(2A_0 + 2\omega A_1 - \omega^2 A_2) \otimes B_1,\qquad L_3=\mathbbm{1} - \frac{1}{3}(2A_0 + 2\omega^2 A_1 - \omega A_2) \otimes B_2.
\end{equation}

Let us then assume that $\ket{\psi}$, $A_x$ and $B_y$ give rise to the maximal quantum value of the inequality. From the above SOS decomposition we 
immediately obtain that $L_i|\psi\rangle=L_i^{\dagger}|\psi\rangle=0$ for every $i=1,2,3$. We then consider a particular pair of equations $L_1\ket{\psi}=0$
and $L_1^{\dagger}\ket{\psi}=0$, which by virtue of the explicit form
of $L_1$ can be rewritten as
\begin{eqnarray}\label{eq1}
    \frac{1}{3}(2A_0 + 2A_1 - A_2) \otimes B_0 |\psi\rangle=|\psi\rangle,\qquad
    \frac{1}{3}(2A_0 + 2A_1 - A_2)^{\dagger} \otimes B_0^{\dagger} |\psi\rangle=|\psi\rangle,
\end{eqnarray}
We then act with $L_1$ on the first equation above which leads us to
\begin{eqnarray}
    \frac{1}{9}(2A_0 + 2A_1 - A_2)^2 \otimes B^2_0 |\psi\rangle=|\psi\rangle,
\end{eqnarray}
Using then the fact that $B_y^{2}=B_y^{\dagger}$ (for $d=3$) and that $B_y$ are unitary, we can rewrite the above as 
\begin{eqnarray}
    \frac{1}{9}[(2A_0 + 2A_1 - A_2)^2\otimes \mathbbm{1} ] |\psi\rangle=(\mathbbm{1}\otimes B_0)|\psi\rangle.
\end{eqnarray}
We can finally use the second equation in \eqref{eq1} to get
\begin{eqnarray}
    [(2A_0 + 2A_1 - A_2)^2\otimes \mathbbm{1} ] |\psi\rangle=3[(2A_0 + 2A_1 - A_2)^{\dagger}\otimes \mathbbm{1} ]|\psi\rangle,
\end{eqnarray}
which because $\ket{\psi}$ is locally full rank, implies 
the following operator equation
\begin{equation}
    (2A_0 + 2A_1 - A_2)^2=3(2A_0 + 2A_1 - A_2)^{\dagger}. 
\end{equation}
This equation can be further rewritten as
\begin{equation}
    2\{A_0, A_1\} - \{A_0, A_2\} - \{A_1, A_2\} = A_0^\dagger + A_1^\dagger - 2A_2^\dagger
\end{equation}
Applying the same steps to the other pairs of equations $L_i|\psi\rangle=L_i^{\dagger}|\psi\rangle=0$, we obtain a system
of equations
\begin{align}
    2\{A_0, A_1\} - \{A_0, A_2\} - \{A_1, A_2\} &= A_0^\dagger + A_1^\dagger - 2A_2^\dagger, \label{eq:1} \\
    2\omega \{A_0, A_1\} - \omega^2 \{A_0, A_2\} -  \{A_1, A_2\} &= A_0^\dagger + \omega^2 A_1^\dagger - 2\omega A_2^\dagger, \label{eq:2} \\
    2\omega^2 \{A_0, A_1\} - \omega \{A_0, A_2\} - \{A_1, A_2\} &= A_0^\dagger + \omega A_1^\dagger - 2\omega^2 A_2^\dagger, \label{eq:3}
\end{align}
which can be significantly simplified to the form 
\begin{align}
\label{ai_relations}
        \{A_1, A_2\} = -A_0^\dagger \\
        \{A_0, A_2\} = -A_1^\dagger \\
        \{A_0, A_1\} = -A_2^\dagger.
\end{align}

Interestingly, the above system was already solved in \cite{Kaniewski_2019}
and the solution is that there exists a unitary operation $U_A:\mathcal{H}_A\to\mathbbm{C}^3\otimes\mathcal{H}_{A''}$ such that
\begin{equation}\label{Aobs}
U_A\, A_{x}\, U_A^{\dagger}=A_{x}'\otimes P_1+ (A_{x}')^T\otimes P_2,    
\end{equation}
where $P_1+P_2=\mathbbm{1}_{A''}$ and $A_{x}'$ are the reference observables given by 
\begin{equation}
    A_0'=Z,\qquad A_1'=X,\qquad A_2'=\omega^2X^2Z^2.
\end{equation}
Notice that $(A_{0}')^T=A_0'$.
%
%
%

Having determined the form of Alice's observables compatible with the maximal violation of the Bell inequality, let us now move on to characterizing Bob's observables. Denoting 
\begin{equation}
\label{c_i's}    C_0=\frac{1}{2}(B_0+B_1+B_2),\;C_1=\frac{1}{2}(B_0+\omega B_1 + \omega^2 B_2),\;C_2=-(B_0+\omega B_1 + \omega^2 B_2),
\end{equation}
one can show from the equations $L_i|\psi\rangle=0$ that 
\begin{equation}
    A_i\otimes C_i|\psi\rangle=|\psi\rangle\qquad (i=0,1,2).
\end{equation}
It is then not difficult to observe that $C^{\dagger}_iC_i=\mathbbm{1}$ 
and $C^3_i=\mathbbm{1}$, which also implies that $C_i^2=C_i^{\dagger}$ 
for every $i=0,1,2$. Then, using the relations \eqref{ai_relations}
obeyed by the $A_i$ observables, one can prove that
\begin{align}
\label{ci_relations}
        \{C_1, C_2\} = -C_0^\dagger \\
        \{C_0, C_2\} = -C_1^\dagger \\
        \{C_0, C_1\} = -C_2^\dagger.
\end{align}
As before, one infers the existence of a unitary operation $U_B$ such that
\begin{eqnarray}
    U_B\, C_{1}\, U_B^\dagger &=& Z^2\otimes\mathbbm{1}\nonumber\\
    U_B\, C_{2}\, U_B^\dagger &=& X\otimes Q_1+X^2\otimes Q_2\nonumber\\
    U_B\, C_{3}\, U_B^\dagger &=& \omega X^2Z\otimes Q_1+\omega^2XZ\otimes Q_2
\end{eqnarray}
with $Q_1+Q_2=\mathbbm{1}_{B}$. This implies that
\begin{equation}\label{Bobs}
    U_B\, B_{y}\, U_B^\dagger = B_y'\otimes Q_{1} + (B_y')^T\otimes Q_{2},
\end{equation}
where
\begin{align}
    B_0'&=\frac{1}{3}\left(2Z^2+2X-\omega X^2Z\right),\\
    B_1'&=\frac{1}{3}\left(2Z^2+2\omega^2X-\omega^2X^2Z\right),\\
    B_2'&=\frac{1}{3}\left(2Z^2+2\omega X- X^2Z\right),
\end{align}

Let us finally prove Eq. (\ref{StateST}). To this end, we plug formulas for
the observables \eqref{Aobs} and \eqref{Bobs} into the Bell operator $\mathcal{F}_3$, which gives us
\begin{equation}
    U_A\otimes U_B\,\mathcal{F}_3\, U^{\dagger}_A\otimes U^{\dagger}_B=\sum_{x,y}W_{xy}\otimes P_x\otimes Q_y,
\end{equation}
where $W_{xy}$ are two-qutrit matrices composed of the reference 
operators $A_x'$ and $B_y'$ and their transpositions. Each $W_{xy}$ can easily be diagonalized and it turns out that only $W_{11}$ and $W_{22}$ contain $\mu=3$ as their eigenvalue, which is the maximal quantum value of $\mathcal{F}_3$. In both cases corresponding eigenspace is 1-dimensional and it is spanned by $\ket{\psi_{+}}=\frac{1}{\sqrt{3}}(\ket{00}+\ket{11}+\ket{22})$.
\end{proof}
\end{document}